\documentclass[sigconf,nonacm]{acmart} 
\usepackage{framed}
\usepackage{amsthm}
\usepackage{booktabs}
\usepackage{multirow}
\usepackage{pifont}
\usepackage{etoolbox}
\usepackage[ruled,vlined]{algorithm2e}
\usepackage{float}      
\usepackage{placeins}   

\newtheorem{proposition}{Proposition}
\AtBeginDocument{%
  }

\begin{document}
\title{Estimating Covariance for Global Minimum Variance Portfolio: A Decision-Focused Learning Approach}

\author{Juchan Kim}
\authornote{Both authors contributed equally to this research.}
\email{joc6021@unist.ac.kr}
\affiliation{%
    \institution{Ulsan National Institute of Science and Technology}
    \city{Ulsan}
    \country{Republic of Korea}
}

\author{Inwoo Tae}
\authornotemark[1]
\email{vitainu0104@unist.ac.kr}
\affiliation{%
   \institution{Ulsan National Institute of Science and Technology}
   \city{Ulsan}
   \country{Republic of Korea}
 }

 \author{Yongjae Lee}
 \authornote{Corresponding Author}
 \email{yongjaelee@unist.ac.kr}
 \affiliation{%
   \institution{Ulsan National Institute of Science and Technology}
   \city{Ulsan}
   \country{Republic of Korea}
 }

\renewcommand{\shortauthors}{Kim et al.}

\begin{abstract}


Portfolio optimization constitutes a cornerstone of risk management by quantifying the risk-return trade-off.
Since it inherently depends on accurate parameter estimation under conditions of future uncertainty, the selection of appropriate input parameters is critical for effective portfolio construction. However, most conventional statistical estimators and machine learning algorithms determine these parameters by minimizing mean‐squared error (MSE), a criterion that can yield suboptimal investment decisions. 



In this paper, we adopt decision‐focused learning (DFL)—an approach that directly optimizes decision quality rather than prediction error such as MSE—to derive the global minimum‐variance portfolio (GMVP). 
Specifically, we theoretically derive the gradient of decision loss using the analytic solution of GMVP and its properties regarding the principal components of itself.
Through extensive empirical evaluation, we show that prediction-focused estimation methods may fail to produce optimal allocations in practice, whereas DFL‐based methods consistently deliver superior decision performance. 
Furthermore, we provide a comprehensive analysis of DFL's mechanism in GMVP construction, focusing on its volatility reduction capability, decision-driving features, and estimation characteristics.


\end{abstract}

\maketitle

\section{Introduction}

Markowitz’s seminal work introduced mean–variance optimization (MVO) \cite{Markowitz52}, which provides a quantitative framework for balancing risk and return and has since become modern portfolio theory's foundation. The most important part in MVO is parameter estimation, due to its high sensitivity \cite{Klein76}. 
Practically, MVO parameters are typically calibrated using an in‐sample period under the assumption of return stationarity, which results in poor performances because the distribution underlying the in‑sample and out‑of‑sample periods often differs. 
Empirical studies have shown that MVO portfolios constructed with estimated parameters can underperform naïve diversification strategies like the equally weighted portfolio, highlighting appropriate parameter estimation \cite{DeMiguel09}.  


Global minimum‐variance portfolio (GMVP) constitutes the least risky portfolio in the efficient frontier, therefore it inherits a key for understanding the modern portfolio theory.
Since GMVP depends solely on the covariance matrix, not on expected returns, its parameter uncertainty becomes much lower than in the MVO framework. 
Moreover, empirical evidence indicates that GMVPs often achieve superior out‐of‐sample return performance compared to market‐cap‐weighted benchmarks \cite{Clarke06,Clarke11} and exhibit lower realized volatility \cite{Bongiorno23}. Because GMVP construction requires estimating only covariances—unlike MVO, which requires both expected returns and covariances—it reduces the dimensionality of the estimation problem. Moreover, while covariance estimators generally achieve higher accuracy than return forecasts \cite{Merton80}, the overall decision quality in MVO remains dominated by errors in expected‐return estimates \cite{Best91,Chopra93}. In summary, GMVPs offer three key advantages over MVO: (1) strong theoretical and empirical support, (2) reduced parameter‐estimation requirements, and (3) decreased sensitivity to estimation error. 


To enhance out‐of‐sample performance—particularly in terms of volatility—several studies have proposed advanced covariance‐estimation for GMVPs \cite{Frahm10,Bongiorno23}. More recently, deep‐learning models have been employed to predict covariance matrices, leveraging their universal‐approximation capabilities in what is known as prediction‐focused learning (PFL) \cite{Petrozziello22}. In the conventional “predict‐then‐optimize” framework, these models are trained to minimize mean‐squared error (MSE) and subsequently used for portfolio construction \cite{lee2024overview}. However, models optimized solely for MSE do not necessarily yield near‐optimal investment decisions \cite{Chung22}. 


Decision‐focused learning (DFL) has emerged as an alternative that directly incorporates the decision objective into the training loss \cite{Elmachtoub22,Mandi22}. While recent work applied DFL to MVO for expected-return estimation \cite{Lee24}, no study has yet explored DFL for covariance estimation in GMVPs. We address this gap by proposing a DFL-based approach with the following contributions:

\begin{itemize}
\item We derive theoretical properties of DFL for GMVPs, focusing on the principal components of the training gradients.
\item Through extensive experiments, we demonstrate that DFL outperforms PFL‐based estimators in constructing GMVPs and reveal the limitations of conventional approaches.
\item We provide evidence that DFL yields stable asset-selection criteria, exhibits robustness over time, correlates with training-period volatility, and offers insight into decision attribution for GMVP construction.
\end{itemize}


\begin{figure*}[htbp]
    \centering
    \includegraphics[width=\textwidth]{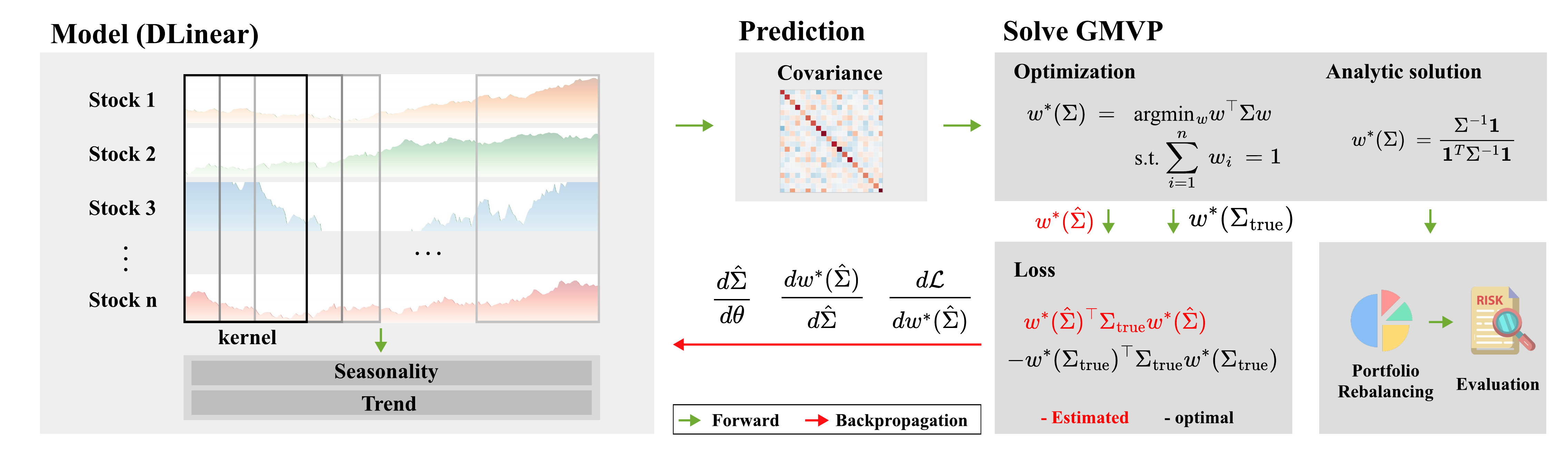}
    \caption{DFL framework for GMVP construction. Historical returns of N assets are processed by DLinear to predict covariance matrix $\Sigma$. The GMVP optimization layer computes portfolio weights $w^*$ using the closed form solution, and the regret loss is applied to minimize the gap between predicted and optimal portfolio volatility for end-to-end training.}
    \Description{model architecture}
    \label{fig:overview}
\end{figure*}

\section{Related Works}

\textbf{Covariance Estimation and Prediction}  
Covariance estimation is a well‐studied problem. In particular, shrinkage estimators \cite{Haff80,LedoitWolf03, LedoitWolf04} combine the sample covariance matrix with a structured target via a convex combination, and this approach has since evolved into nonlinear shrinkage methods \cite{LedoitWolf12}; both classes of estimators provide flexible, general‐purpose covariance estimates.
Shrinkage estimators are typically optimized under the Frobenius norm, which does not directly align with portfolio‐optimization objectives, and may therefore fail to capture the covariance structures most relevant for risk minimization. To address this, some studies replace the Frobenius‐norm loss with a task‐specific loss—namely the portfolio variance—thereby emphasizing out‐of‐sample performance by minimizing portfolio volatility subject to weight and eigenvalue constraints \cite{Bongiorno23,MörstedtLutzNeumann24}. Other work has developed covariance estimators tailored to the GMVP setting—employing Bayesian techniques or insights from random‐matrix theory—and has established their statistical guarantees and empirical performance \cite{Frahm10,BodnarMazurOkhrin17,BodnarParolyaSchmid18}. While these methods offer strong theoretical properties, they remain limited to prespecified parametric forms, which can restrict expressiveness.

\vspace{0.5\baselineskip}
\noindent\textbf{Decision‐Focused Learning}  
Decision‐focused learning (DFL) has recently emerged as an alternative to prediction‐focused learning (PFL), often producing superior decision quality in downstream tasks \cite{MandiKotaryBerden24}. Several authors provide theoretical analyses of DFL by deriving gradient expressions for optimization layers with respect to model parameters. In particular, \cite{AmosKolter17, AgrawalEtAl19} introduce differentiation rules for optimization problems as differentiable layers. When the optimization map is non-differentiable, surrogate losses with well‐defined subgradients have been proposed to enable end‐to‐end training \cite{Elmachtoub22,ShahEtAl22}. DFL has been applied across diverse domains, including route optimization \cite{FerberEtAl23}, communication networks \cite{ChaiEtAl22}, and portfolio optimization \cite{ButlerKwon23}. Within finance, DFL has been used to enhance model robustness \cite{CostaIyengar23} and to investigate the role of expected‐return estimation in MVO \cite{Lee24}. However, to our knowledge, no prior work has applied DFL to covariance estimation for the global minimum‐variance portfolio, which is the focus of our study.

\section{Methodology}

In this section, we present our overall approach. Figure~\ref{fig:overview} illustrates the complete model architecture. We focus on the unconstrained global minimum‐variance portfolio (GMVP), which admits a closed‐form solution and thus facilitates both theoretical analysis and efficient computation.

\subsection{Notation and Problem Setup}

Let \(X\in\mathbb{R}^{T\times N}\) denote the matrix of asset returns for \(N\) assets over \(T\) time steps, and let
\[
X_{t_1:t_2}\,\in\,\mathbb{R}^{(t_2 - t_1 + 1)\times N}
\]
be the submatrix of returns from time \(t_1\) to \(t_2\). We write
\(\Sigma_{t_1:t_2}\) for the covariance matrix of \(X_{t_1:t_2}\).  

In a traditional time‐series forecasting task, one predicts the next \(\delta_{\mathrm{out}}\) observations using the previous \(\delta_{\mathrm{in}}\) observations. Accordingly, we define a prediction model
\[
f_{\theta} : X_{t-\delta_{\mathrm{in}}+1:t}\;\mapsto\;\hat{\Sigma}_{t+1:t+\delta_{\mathrm{out}}},
\]
where \(\hat{\Sigma}_{t+1:t+\delta_{\mathrm{out}}}\) is the model’s estimate of the future covariance. We denote the resulting unconstrained GMVP weights by

\[
w^*(\hat{\Sigma})
=\;\arg\min_{w}\;w^\top \hat{\Sigma}\,w
\quad\text{s.t.}\;\mathbf{1}^\top w = 1
\;\Longrightarrow\;
w^*(\hat{\Sigma})
=\;\frac{\hat{\Sigma}^{-1}\mathbf{1}}{\mathbf{1}^\top \hat{\Sigma}^{-1}\mathbf{1}}.
\]

\subsection{Covariance Prediction model}

To capture the temporal dynamics in \(X\), we employ the DLinear backbone \cite{ZengEtAl23}, which has demonstrated strong performance on a variety of time‐series tasks despite its architectural simplicity. Because a covariance matrix must be symmetric and positive semi‐definite, our model predicts the entries of a lower‐triangular matrix \(L\), and we reconstruct the covariance as
\[
\hat{\Sigma} \;=\; L\,L^\top,
\]
ensuring symmetry and positive semi‐definiteness by construction. To improve numerical stability when inverting \(\hat{\Sigma}\), we perform a truncated spectral reconstruction: we retain only eigenpairs \((\lambda_i, v_i)\) satisfying \(\lambda_i \ge \varepsilon\,\lambda_{\max}\).

\subsection{Decision‐Focused Learning with Regret Loss}

Given \(\hat{\Sigma}\), we compute the unconstrained GMVP weights \(w^*(\hat{\Sigma})\). We then assess decision quality via the \emph{regret} loss:
\[
L(\theta)
=\,w^*(\hat{\Sigma})^\top\,\Sigma_{\mathrm{true}}\,w^*(\hat{\Sigma})
\;-\;
w^*(\Sigma_{\mathrm{true}})^\top\,\Sigma_{\mathrm{true}}\,w^*(\Sigma_{\mathrm{true}}),
\]
where \(\Sigma_{\mathrm{true}} = \Sigma_{t+1:t+\delta_{\mathrm{out}}}\) is the oracle covariance. We update \(\theta\) by gradient descent, computing
\[
\frac{dL}{d\theta}
= \frac{dL}{d w^*(\hat{\Sigma})}\,\frac{\partial w^*(\hat{\Sigma})}{\partial \hat{\Sigma}}\,\frac{\partial \hat{\Sigma}}{\partial \theta}.
\]
Closed‐form expressions for the first two factors are
\begin{align*}
\frac{dL}{dw}
&= 2\,\Sigma_{\mathrm{true}}\,w^*(\hat{\Sigma}),\\
\frac{\partial w^*(\hat{\Sigma})}{\partial \hat{\Sigma}}
&= -\,w^\top \otimes \bigl(\hat{\Sigma}^{-1}\bigr)^\top
   + D\,\mathrm{vec}(w)\,(w\otimes w)^\top.
\end{align*}

where \(\otimes\) denotes the Kronecker product, \(\mathbf{1}\in\mathbb{R}^N\) is the all‐ones vector, \(D = \mathbf{1}^\top\hat{\Sigma}^{-1}\mathbf{1}\), and \(w = w^*(\hat{\Sigma})\). The term \(\partial\hat{\Sigma}/\partial\theta\) is obtained via automatic differentiation. This completes the end‐to‐end decision‐focused learning pipeline for GMVP construction.

\section{Theoretical Analysis}
\begin{table*}[htbp]
\centering
\caption{Average annualized volatility in each portfolio. All models use $\delta_{\text{in}}=21$ days for estimation. Historical: historical covariances with lookback $\delta_{in}$; EW: Equal-weighted portfolio; LW-D: Ledoit-Wolf shrinkage estimator with diagonal shrinkage target \cite{LedoitWolf04}; LW-CC: Ledoit-Wolf shrinkage estimator with constant correlation model \cite{LedoitWolf03}; OAS: Oracle approximating shrinkage \cite{ChenEtAl10}; PFL: Prediction-focused learning; Ours (DFL): Decision-focused learning (mean and std over 5 seeds). Bold values indicate best performing method.}
\label{tab:avg_ann_vol}
\resizebox{\textwidth}{!}{%
\begin{tabular}{l*{15}{c}}
\toprule
& \multicolumn{5}{c}{S\&P} & \multicolumn{5}{c}{Industry} & \multicolumn{5}{c}{Dow Jones} \\
\cmidrule(lr){2-6} \cmidrule(lr){7-11} \cmidrule(lr){12-16}
Model & $\delta_{\text{out}}=5$ & $\delta_{\text{out}}=21$ & $\delta_{\text{out}}=63$ & $\delta_{\text{out}}=126$ & $\delta_{\text{out}}=252$ & $\delta_{\text{out}}=5$ & $\delta_{\text{out}}=21$ & $\delta_{\text{out}}=63$ & $\delta_{\text{out}}=126$ & $\delta_{\text{out}}=252$ & $\delta_{\text{out}}=5$ & $\delta_{\text{out}}=21$ & $\delta_{\text{out}}=63$ & $\delta_{\text{out}}=126$ & $\delta_{\text{out}}=252$ \\
\midrule
EW         & 0.1489 & 0.1550 & 0.1597 & 0.1611 & 0.1772 & 0.1670 & 0.1732 & 0.1765 & 0.1775 & 0.1917 & 0.1438 & 0.1504 & 0.1558 & 0.1565 & 0.1687 \\
\addlinespace
Historical      & 0.1249 & 0.1382 & 0.1471 & 0.1466 & 0.1416 & 0.1427 & 0.1468 & 0.1601 & 0.1620 & 0.1600 & 0.1531 & 0.1768 & 0.1791 & 0.1858 & 0.1944 \\
\addlinespace
LW-D & 0.1221 & 0.1334 & 0.1407 & 0.1416 & 0.1401 & 0.1329 & 0.1355 & 0.1435 & 0.1465 & 0.1432 & 0.1252 & 0.1408 & 0.1419 & 0.1520 & 0.1584 \\
\addlinespace
LW-CC         & 0.1289 & 0.1431 & 0.1437 & 0.1495 & 0.1396 & 0.1402 & 0.1503 & 0.1630 & 0.1707 & 0.1648 & 0.1491 & 0.1707 & 0.1777 & 0.1796 & 0.1895 \\
\addlinespace
OAS        & 0.1206 & 0.1317 & 0.1383 & 0.1395 & 0.1399 & 0.1254 & 0.1301 & 0.1364 & 0.1383 & 0.1403 & 0.1219 & 0.1370 & 0.1370 & 0.1433 & 0.1549 \\
\midrule
PFL & \begin{tabular}[c]{@{}c@{}}0.1487\\{\scriptsize (0.0002)}\end{tabular} & \begin{tabular}[c]{@{}c@{}}0.1549\\{\scriptsize (0.0002)}\end{tabular} & \begin{tabular}[c]{@{}c@{}}0.1595\\{\scriptsize (0.0003)}\end{tabular} & \begin{tabular}[c]{@{}c@{}}0.1607\\{\scriptsize (0.0005)}\end{tabular} & \begin{tabular}[c]{@{}c@{}}0.1773\\{\scriptsize (0.0007)}\end{tabular} & \begin{tabular}[c]{@{}c@{}}0.1670\\{\scriptsize (0.0002)}\end{tabular} & \begin{tabular}[c]{@{}c@{}}0.1730\\{\scriptsize (0.0002)}\end{tabular} & \begin{tabular}[c]{@{}c@{}}0.1764\\{\scriptsize (0.0002)}\end{tabular} & \begin{tabular}[c]{@{}c@{}}0.1775\\{\scriptsize (0.0002)}\end{tabular} & \begin{tabular}[c]{@{}c@{}}0.1917\\{\scriptsize (0.0005)}\end{tabular} & \begin{tabular}[c]{@{}c@{}}0.1438\\{\scriptsize (0.0006)}\end{tabular} & \begin{tabular}[c]{@{}c@{}}0.1501\\{\scriptsize (0.0005)}\end{tabular} & \begin{tabular}[c]{@{}c@{}}0.1557\\{\scriptsize (0.0009)}\end{tabular} & \begin{tabular}[c]{@{}c@{}}0.1568\\{\scriptsize (0.0007)}\end{tabular} & \begin{tabular}[c]{@{}c@{}}0.1695\\{\scriptsize (0.0015)}\end{tabular} \\
\addlinespace
Ours (DFL) & \begin{tabular}[c]{@{}c@{}}\textbf{0.1154}\\{\scriptsize (0.0008)}\end{tabular} & \begin{tabular}[c]{@{}c@{}}\textbf{0.1219}\\{\scriptsize (0.0005)}\end{tabular} & \begin{tabular}[c]{@{}c@{}}\textbf{0.1290}\\{\scriptsize (0.0012)}\end{tabular} & \begin{tabular}[c]{@{}c@{}}\textbf{0.1287}\\{\scriptsize (0.0015)}\end{tabular} & \begin{tabular}[c]{@{}c@{}}\textbf{0.1319}\\{\scriptsize (0.0013)}\end{tabular} & \begin{tabular}[c]{@{}c@{}}\textbf{0.1154}\\{\scriptsize (0.0005)}\end{tabular} & \begin{tabular}[c]{@{}c@{}}\textbf{0.1223}\\{\scriptsize (0.0009)}\end{tabular} & \begin{tabular}[c]{@{}c@{}}\textbf{0.1282}\\{\scriptsize (0.0004)}\end{tabular} & \begin{tabular}[c]{@{}c@{}}\textbf{0.1281}\\{\scriptsize (0.0014)}\end{tabular} & \begin{tabular}[c]{@{}c@{}}\textbf{0.1379}\\{\scriptsize (0.0017)}\end{tabular} & \begin{tabular}[c]{@{}c@{}}\textbf{0.1170}\\{\scriptsize (0.0006)}\end{tabular} & \begin{tabular}[c]{@{}c@{}}\textbf{0.1230}\\{\scriptsize (0.0006)}\end{tabular} & \begin{tabular}[c]{@{}c@{}}\textbf{0.1277}\\{\scriptsize (0.0004)}\end{tabular} & \begin{tabular}[c]{@{}c@{}}\textbf{0.1283}\\{\scriptsize (0.0006)}\end{tabular} & \begin{tabular}[c]{@{}c@{}}\textbf{0.1372}\\{\scriptsize (0.0007)}\end{tabular} \\
\bottomrule
\end{tabular}%
}
\end{table*}


\setlength{\FrameRule}{0.8pt}   
\setlength{\FrameSep}{2pt}      

In this section, we theoretically characterize the decision-aware gradient by deriving conditions for the singular vectors of
\(\tfrac{\partial w}{\partial \hat{\Sigma}}\) and
\(\tfrac{\partial L}{\partial \hat{\Sigma}}\)
 to remain \(\hat{\Sigma}\)-invariant, as stated in proposition~\ref{prop_1}.

\begin{proposition}\label{prop_1}
Let 
$J := \frac{\partial w^*(\hat{\Sigma})}{\partial \hat{\Sigma}}$. Define
\vspace{0.2em}
\begin{align*}
&JJ^\top = \{(w^\top w) (\hat{\Sigma}^{-1})^2\} + \{D^2(w^\top w)^2 (w w^\top) \notag\\
&\phantom{JJ^\top =} - D(w^\top w)[\hat{\Sigma}^{-1}(w w^\top)+w(w^\top \hat{\Sigma}^{-1})]\} =: S_{1} + R_{1},
\end{align*}
\vspace{-1.2em}
\begin{align*}
&J^\top J = (w w^\top)\otimes (\hat{\Sigma}^{-1})^2 + \{D^2(w^\top w) (w \otimes w)(w \otimes w)^{\top} \notag\\
&\phantom{J^\top J =} - D(w\otimes\hat{\Sigma}^{-1})\mathrm{vec}(w)(w\otimes w)^{\top} \notag\\
&\phantom{J^\top J =} - D(w\otimes w)\mathrm(w)^{\top}(w\otimes\hat{\Sigma}^{-1})^{\top} \} =: S_{2} + R_{2}.
\end{align*}

Assume that 
\begin{itemize}
\item $\text{ker}(\hat{\Sigma}^{-1})^2 =\text{span}(\{w, \hat{\Sigma}^{-1}w\})$,
\item $\text{ker}((ww^{\top})\otimes(\hat{\Sigma}^{-1})^2) =\text{span}(\{(w\otimes w), (w \otimes \hat{\Sigma}^{-1}w)\})$,
\item $R_1-I_n$ and $R_2-I_{n^2}$ are invertible.
\end{itemize}

Then, there exist at least two $\hat{\Sigma}$-invariant singular vectors for $JJ^\top$ and $J^\top J$ with identical eigenvalues which are in $\text{span}(\{w, \hat{\Sigma}^{-1}w\})$ and $\text{span}(\{(w\otimes w), (w \otimes \hat{\Sigma}^{-1}w\})$, respectively.
Here, we denoted $I_k$ as $k \times k$ identity matrix.
\end{proposition}

\begin{proof}
Let $\text{col}(\cdot)$ denote the column space of a matrix and define $a = w^\top w$, $b = w^\top \hat{\Sigma}^{-1}w$, and $c = w^\top (\hat{\Sigma}^{-1})^2 w$.

\vspace{0.4em}
\textbf{Claim 1:} {\itshape $JJ^\top$ has eigenvectors in $\text{span}(\{w, \hat{\Sigma}^{-1}w\})$.}

Since $\text{col}(R_1) = \text{span}(\{w, \hat{\Sigma}^{-1}w\})$ with $\text{rank}(R_1) \leq 2$, the matrix representation of $R_1$ with basis $\mathcal{B}_1 = \{w, \hat{\Sigma}^{-1}w\}$ is equivalent to:

\begin{align*}
    [R_1]_{\mathcal{B}_1} = \begin{pmatrix}
    D^2a^3-Dab & D^2a^2b-Dac \\
    -Da^2 & -Dab
    \end{pmatrix}.
\end{align*}

Matrix determinant lemma gives
\begin{align*}
\det(JJ^\top-\lambda I_n) &= \det([R_1]_{\mathcal{B}_1} - \lambda I_2) \cdot \det(I_n + (R_1-\lambda I_n)^{-1}S_1),
\end{align*}
where $R_1 = U_1V_1^\top$ with $U_1 = [w, \hat{\Sigma}^{-1}w]$.
    
Thus, eigenvalues of $JJ^\top$ include those of $[R_1]_{\mathcal{B}_1}$. For eigenvector $\textbf{x} \in \text{span}(\{w, \hat{\Sigma}^{-1}w\})$ of $R_1$ with eigenvalue $\lambda_1$, using the fact $\text{ker}(S_1) = \text{span}(\{w, \hat{\Sigma}^{-1}w\})$, we have a desired result as follows:
\begin{align*}
(JJ^\top)\textbf{x} = (S_1+R_1)\textbf{x} = 0 + \lambda_1\textbf{x} = \lambda_1\textbf{x}.
\end{align*}
    
\textbf{Claim 2:} {\itshape $J^\top J$ has eigenvectors in $\text{span}(\{w \otimes w, w \otimes \hat{\Sigma}^{-1}w\})$.
    }
Similar with claim 1, from $\mathcal{B}_2 = \{w \otimes \hat{\Sigma}^{-1}w, w \otimes w\}$, we have:
\begin{align*}
[R_2]_{\mathcal{B}_2} = \begin{pmatrix}
-Dab & -Da^2 \\
D^2a^2b-Dac & D^2a^3-Dab
\end{pmatrix}.
\label{formula:R_2_rep}
\end{align*}
For eigen-pairs $(\textbf{y}, \lambda_2)$ of $[R_2]_{\mathcal{B}_2}$ where $\textbf{y} \in \text{span}(\{w \otimes \hat{\Sigma}^{-1}w, w \otimes w\})$, $(J^{\top} J)\textbf{y}=\lambda_{2} \textbf{y}$ implies that $\textbf{y}$ is an eigenvector of $J^\top J$.
    
\vspace{0.6em}
\textbf{Claim 3:} {\itshape Eigenvectors of $[R_1]_{\mathcal{B}_1}$ and $[R_2]_{\mathcal{B}_2}$ share identical eigenvalues.}

Observe that
\begin{align*}
\det([R_1]_{\mathcal{B}_1} - \lambda I_2) = \det([R_2]_{\mathcal{B}_2} - \lambda I_2),
\end{align*}
establishing $[R_1]_{\mathcal{B}_1}$ and $[R_2]_{\mathcal{B}_2}$ have identical eigenvalues. 
\end{proof}

It states there exist $\hat{\Sigma}$-invariant singular vectors such that they share same eigenvalues for some conditions. Therefore, it implies that $\frac{dw}{d\hat{\Sigma}}$ has principal components in subspaces, inducing the closed-form of singular vectors of $\frac{dw}{d\hat{\Sigma}}$ stated in the next proposition.

\vspace{0.2em}
\begin{proposition}\label{prop_2}
Under the assumptions of Proposition~\ref{prop_1}, Let $\lambda$ be shared eigenvalue $R_1$ and $R_2$ and let 
\(\mathbf{x}(\lambda)\) and \(\mathbf{y}(\lambda)\)
be the corresponding eigenvectors of \(R_1\) and \(R_2\), respectively. Then,
\vspace{-0.2em}
\begin{align}
    \textbf{x}(\lambda) &= w+\frac{\lambda - (D^2 a^3-Dab)}{D^2a^2 b - ac}\hat{\Sigma}^{-1}w, \\
    \textbf{y}(\lambda) &= (w \otimes \hat{\Sigma}^{-1}w)+\frac{\lambda - (D^2 a^3-Dab)}{D^2a^2 b - ac}(w \otimes w).
\end{align}
\label{formula:eigenvecs}
\end{proposition}

\begin{proof}
From the characteristic equaiton of $[R_1]_{\mathcal{B}_1}{}$, we have
    
\begin{align}
    ([R_1]_{\mathcal{B}_1}-\lambda I_{n}) \begin{pmatrix}
\alpha \\
\beta
\end{pmatrix}=0
\end{align}

which implies $(D^{2}a^{3}-Dab-\lambda)\alpha+(Da^{2}b-ac)\beta = 0$.
Hence, we derive $\frac{\beta}{\alpha} = \frac{\lambda - D^{2}a^{3}-Dab}{Da^{2}b-ac}$ and the eigenvector $x(\lambda)=\alpha w + \beta \Sigma^{-1}w$ which is a desired result. 
Since we showed that $R_{1}$ and $R_{2}$ have identical characteristic equation in proposition \ref{prop_1}, we have $\textbf{y}(\lambda) = \alpha (w \otimes \hat{\Sigma}^{-1}w) + \beta (w \otimes w)$.
\end{proof}

\vspace{0.4em}
Using results of proposition \ref{prop_1} and \ref{prop_2}, we can relate principal components of the decision‐loss gradient \(\tfrac{\partial L}{\partial \hat{\Sigma}}\) which is discussed in proposition \ref{prop_3}.

\vspace{0.2em}
\begin{proposition}\label{prop_3}
Let $L$ be a regret loss and $F := \frac{\partial L}{\partial w} = 2\Sigma_{\mathrm{true}} w^*(\hat{\Sigma}) \in \mathbb{R}^{n}$ be a gradient of $L$ with respect to portfolio weight $w$. With assumptions in proposition \ref{prop_1}, $FJ \in \mathbb{R}^{1 \times n^{2}}$ has two $\Sigma$-invariant singular vectors which is form of (\ref{formula:eigenvecs}) 
where $\lambda$ is solution of characteristic equation of $[R_2]_{\mathcal{B}_{2}}$. 
Moreover, its singular value is equal to $q(\lambda)=\lambda(2\Sigma_{true}w)^{\top}\textbf{x}(\lambda)$.
\end{proposition}

\begin{proof}
Denote singular value decomposition of $J$ as $U\Lambda V^{\top}$ where $U\in\mathbb{R}^{n\times n}$, $\Lambda \in \mathbb{R}^{n \times n^{2}}$ and $V \in \mathbb{R}^{n^{2} \times n^{2}}$. From proposition \ref{prop_1} and \ref{prop_2}, $\textbf{x}(\lambda)$ and $\textbf{y}(\lambda)$ are in the same index of column vector of $U$ and $V$, respectively. By defining the index as $i\in\{1,2,\cdots,n^{2}\}$, therefore, we have
\begin{align*}
    (FJ)_{i} &= \sqrt{\lambda}\{(2\Sigma_{true}w^{*})^{\top} \textbf{x}(\lambda)\}\textbf{y}(\lambda) 
\end{align*}
which is a desired result.
\end{proof}

\vspace{0.2em}                
From proposition~\ref{prop_3}, we see that:
\begin{enumerate}
  \item The direction \(w\otimes(\Sigma^{-1}w)\) captures interactions between the GMVP and its risk‐adjusted version, akin to a risk‐parity structure.
  \item The direction \(w\otimes w\) emphasizes assets with large absolute weights (\(\lvert w_i\rvert \ge 1\)) and attenuates those within the budget constraint (\(\lvert w_i\rvert \le 1\)).
\end{enumerate}

\paragraph{Remark.}
We assume that the kernels of \((\Sigma^{-1})^2\) and \((w\,w^\top)\otimes(\Sigma^{-1})^2\) are non‐trivial. In practice, this condition may fail, which constitutes a limitation of our analysis.

\section{Experiments}

\subsection{Experimental details}

\textbf{Data} We conduct experiments on three asset universes. First, we employ the 49 industry portfolios from the Kenneth R. French data library\footnote{\url{https://mba.tuck.dartmouth.edu/pages/faculty/ken.french/data_library.html}}, hereafter denoted “49 Industry.” Second, we select the top 100 S\&P 500 stocks by market capitalization as of April 1, 2025 (denoted “S\&P 100”). Third, we use the constituents of the Dow Jones 30 index (“Dow 30”). For each universe, we compute daily returns from January 1, 2010, to December 31, 2024, and estimate the covariance matrix from these returns. In the S\&P 100 and Dow 30 universes, we exclude any stock with missing return data or with pairwise correlation \(\ge0.95\), yielding final universes of 100 and 18 stocks, respectively. We split the data into training, validation, and test sets in a 6:2:2 ratio, yielding training set from January 5, 2010 to May 24, 2018; validation set from May 25, 2018 to March 12, 2021; and a test set from March 15, 2021 to December 29, 2023. 

\begin{figure}[h]
    \centering
    \includegraphics[width=0.88\columnwidth]{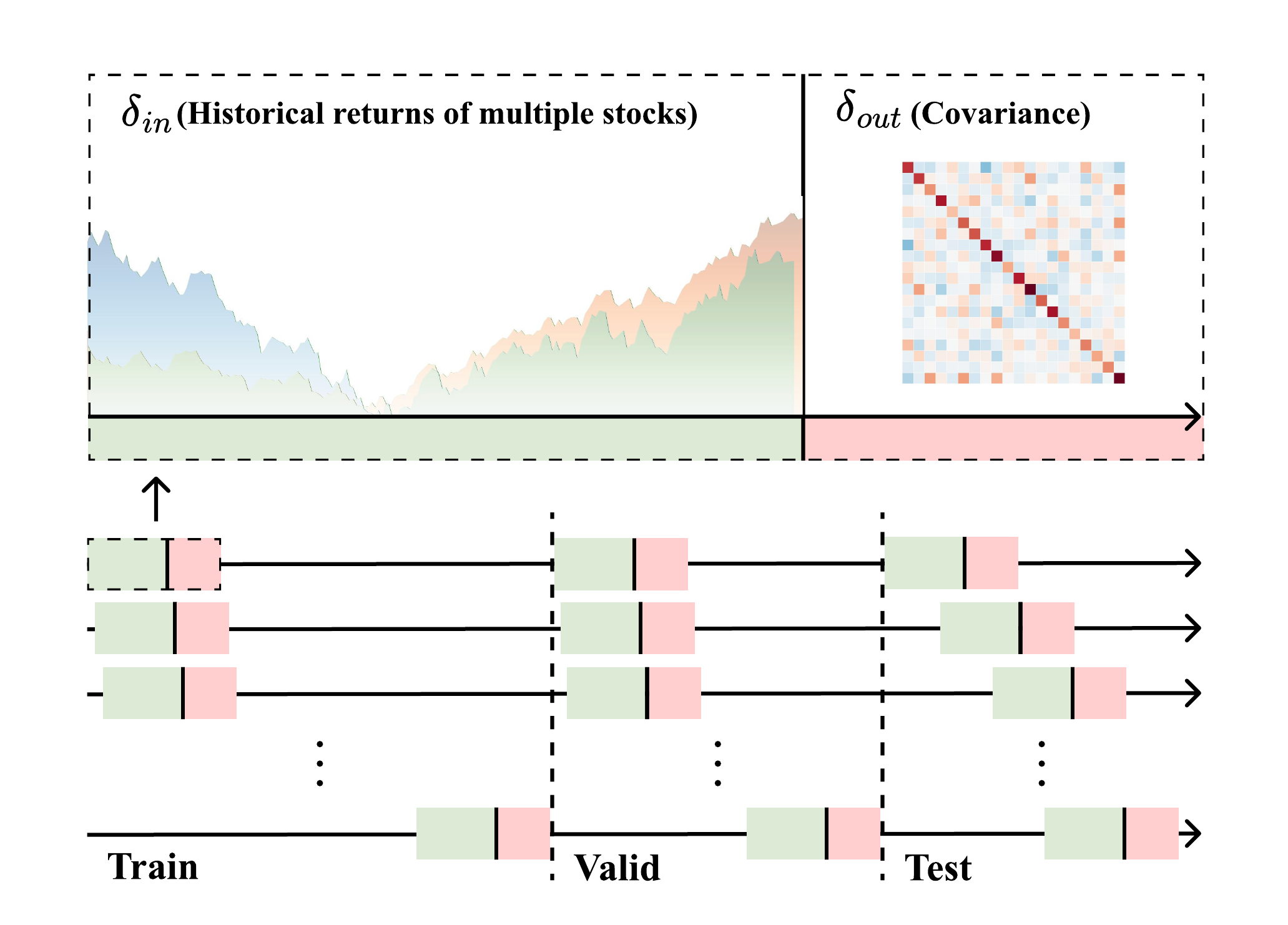}
    \caption{Rolling window approach for time series validation. Train/valid sets use one-day rolling windows, while the test set employs $\delta_{out}$-days rolling predictions.}
    \Description{Ann vol plot}
    \label{fig:rolling}
\end{figure}

\vspace{0.5\baselineskip}
\noindent\textbf{Hyperparameters} We grid-search all hyperparameters. DLinear's kernel size for computing moving averages and residuals is $\max\{5,\min\{\lfloor\delta_{\mathrm{in}}/3\rfloor,50\}\}$, with hidden dimension $h_{\mathrm{dim}}=128$ throughout. We train with the Adam optimizer \cite{KingmaBa14}. Table~\ref{table:hyperparameters} summarizes the learning‐rate and batch‐size settings for each \((\delta_{\mathrm{in}},\delta_{\mathrm{out}})\) configuration. All models are trained for up to 50 epochs, with early stopping when validation loss shows no improvement for 7 epochs.  

\begin{table}[ht]
  \centering
  \small
  \caption{Hyperparameter settings in (\text{learning rate}, \text{batch}) pairs for varying $\delta_{in}$ and $\delta_{out}$}
  \label{tab:delta}
  \begin{tabular}{c|ccccc}
    \toprule
    & \multicolumn{5}{c}{$\delta_{in}$} \\
    \cmidrule(lr){2-6}
    $\delta_{out}$ 
      & 5      & 21     & 63     & 126    & 252    \\
    \midrule
    5   & ($10^{-4}$,\,64) & ($10^{-5}$,\,16) & ($10^{-4}$,\,16) & ($10^{-4}$,\,32) & ($10^{-5}$,\,16) \\
    21  & ($10^{-5}$,\,32) & ($10^{-5}$,\,16) & ($10^{-5}$,\,64) & ($10^{-4}$,\,32) & ($10^{-4}$,\,32) \\
    63  & ($10^{-4}$,\,64) & ($10^{-4}$,\,32) & ($10^{-4}$,\,64) & ($10^{-4}$,\,64) & ($10^{-4}$,\,32) \\
    126 & ($10^{-5}$,\,64) & ($10^{-4}$,\,32) & ($10^{-4}$,\,64) & ($10^{-4}$,\,64) & ($10^{-5}$,\,16) \\
    252 & ($10^{-3}$,\,32) & ($10^{-5}$,\,16) & ($10^{-3}$,\,16) & ($10^{-4}$,\,16) & ($10^{-5}$,\,16) \\
    \bottomrule
  \end{tabular}
  \label{table:hyperparameters}
\end{table}

\begin{figure*}[htbp]
    \centering
    \includegraphics[width=0.9\textwidth]{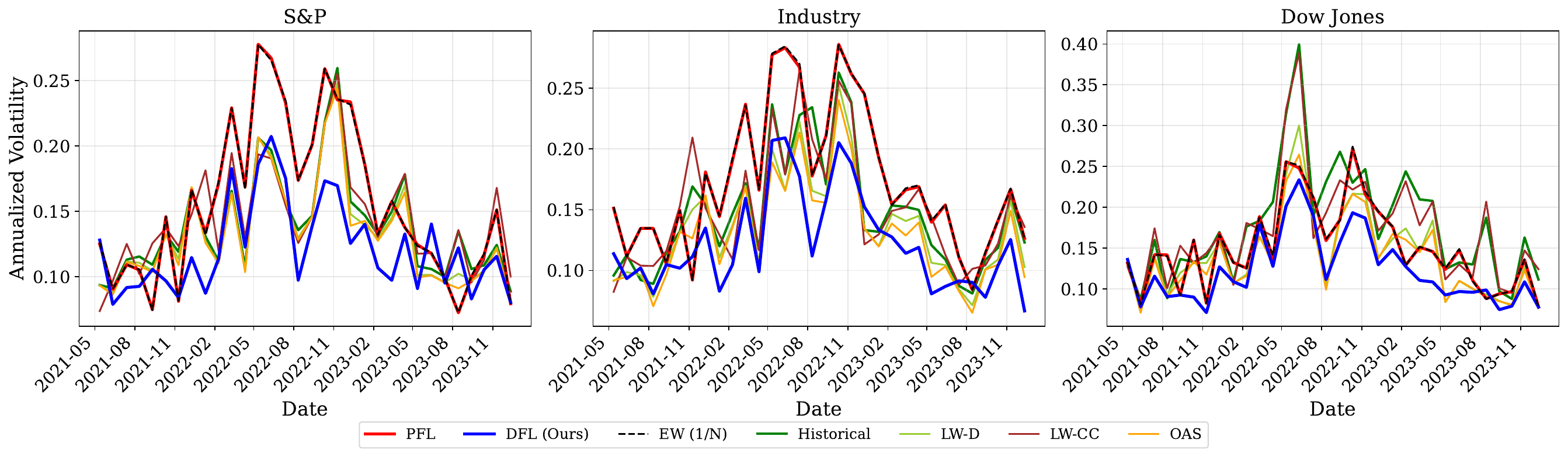}
    \caption{Annualized volatility of each model's portfolio with $\delta_{in}=\delta_{out}=21$ from the left to the right, respectively.}
    \Description{Ann vol plot}
    \label{fig:ann_vol}
\end{figure*}

\subsection{Results}

We evaluate each model’s performance under a buy‐and‐hold strategy: at time \(t\), we compute portfolio weights \(w\) from the predicted covariance \(\hat{\Sigma}_{t:t+\delta_{\mathrm{out}}}\), hold \(w\) until \(t+\delta_{\mathrm{out}}\), then rebalance and repeat. We illustrate the difference of rolling method between training and test period in Figure \ref{fig:rolling}. Performance is measured by the average annualized volatility of realized portfolio returns. We compare our DFL approach against the following benchmarks:

\begin{description}
  \item[\textbf{EW}] Equally weighted portfolio.
  \item[\textbf{Historical}] GMVP using the sample covariance over the past \(\delta_{\mathrm{in}}\) returns.
  \item[\textbf{LW‑D}] GMVP with Ledoit–Wolf shrinkage toward a scaled identity matrix \cite{LedoitWolf04}.
  \item[\textbf{LW‑CC}] GMVP with Ledoit–Wolf constant‐correlation shrinkage \cite{LedoitWolf03}.
  \item[\textbf{OAS}] GMVP with the Oracle Approximating Shrinkage estimator \cite{ChenEtAl10}.
  \item[\textbf{PFL}] GMVP using covariance forecasts from the DLinear module, trained to minimize mean‐squared error (MSE).
\end{description}


Table~\ref{tab:avg_ann_vol} reports the average annualized volatility over the test period and Figure~\ref{fig:ann_vol} shows annualized volatility during each test period. Conventional shrinkage and historical estimators— which optimize the Frobenius‐norm loss—fail to match the decision quality of DFL, since minimizing Frobenius error does not directly translate to low portfolio volatility. The PFL model performs comparably to the equally weighted portfolio, underscoring that MSE‐based covariance forecasting can yield nearly constant off‐diagonal estimates and thus an approximately equal‐weight allocation.

\begin{figure}[htbp]
    \vspace{-0.2cm}
    \centering
    \includegraphics[width=\columnwidth]{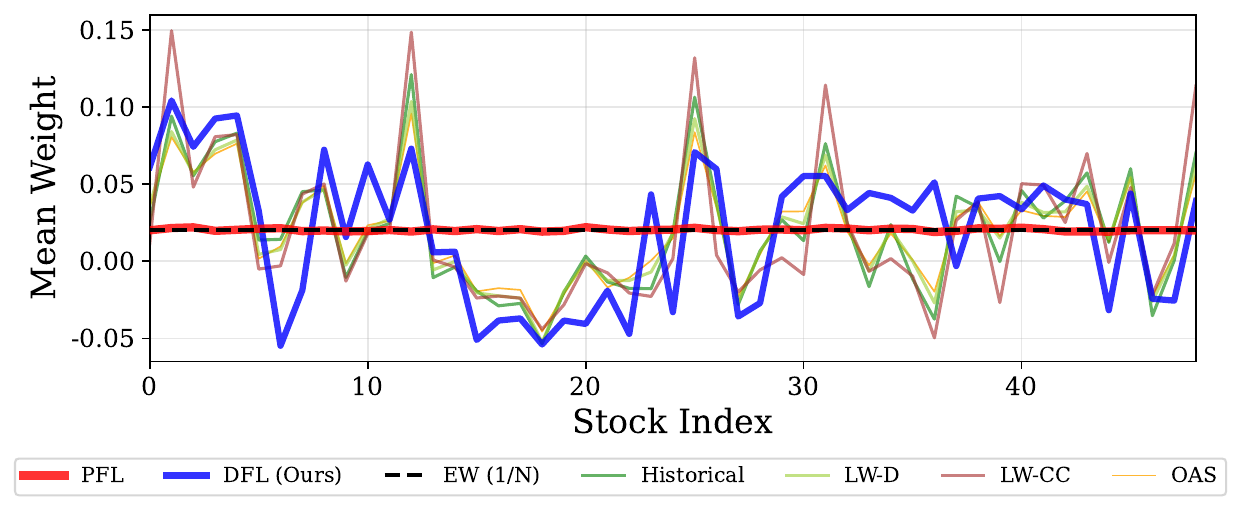}
    \caption{Average portfolio weights in the Industry portfolio dataset during whole test period with $\delta_{in}=\delta_{out}=21$.}
    \Description{model architecture}
    \label{fig:ew_pfl}
    \vspace{-0.2cm}
\end{figure}


Figure~\ref{fig:ew_pfl} depicts the average sector weights for the 49‐Industry dataset over the test period. The PFL allocations closely mirror those of the EW benchmark, whereas DFL and other methods exhibit distinct exposure profiles. This phenomenon arises because training phase with MSE prioritizes reducing large diagonal errors in the covariance matrix, leading the network to underfit off‐diagonal terms and effectively produce nearly uniform correlations. It supports that PFL based on minimizing MSE might not be helpful to the parameter estimation in decision making, wherea DFL are successfully achieves great decision making. Similar results would be observed for S\&P and Dow Jones datasets, reported in the figure \ref{fig:ew_pfl_appendix}.

\begin{figure}[h]
    \centering
    \includegraphics[width=\columnwidth]{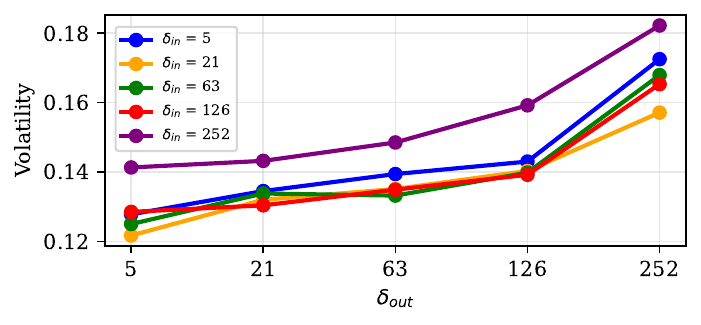}
    \caption{Annualized volatility of DFL when $\delta_{in}$ and $\delta_{out}$ have different values in the Industry dataset. We set $\delta_{in}, \delta_{out}\in \{5,21,63,126,252\}$.}
    \Description{Ann vol plot}
    \label{fig:delta_ablation}
\end{figure}



Figure~\ref{fig:delta_ablation} presents an ablation study over various \((\delta_{\mathrm{in}},\delta_{\mathrm{out}})\) pairs. As \(\delta_{\mathrm{out}}\) increases for a fixed \(\delta_{\mathrm{in}}\), DFL’s volatility rises, reflecting the greater uncertainty in longer‐horizon forecasts. Conversely, for a fixed \(\delta_{\mathrm{out}}\), performance varies only modestly with \(\delta_{\mathrm{in}}\), indicating robustness to the length of the input window but also highlighting the importance of selecting an appropriate length of input sequence.

\subsection{Empirical properties of DFL}
\begin{figure*}[h]
    \centering
    \includegraphics[width=\textwidth]{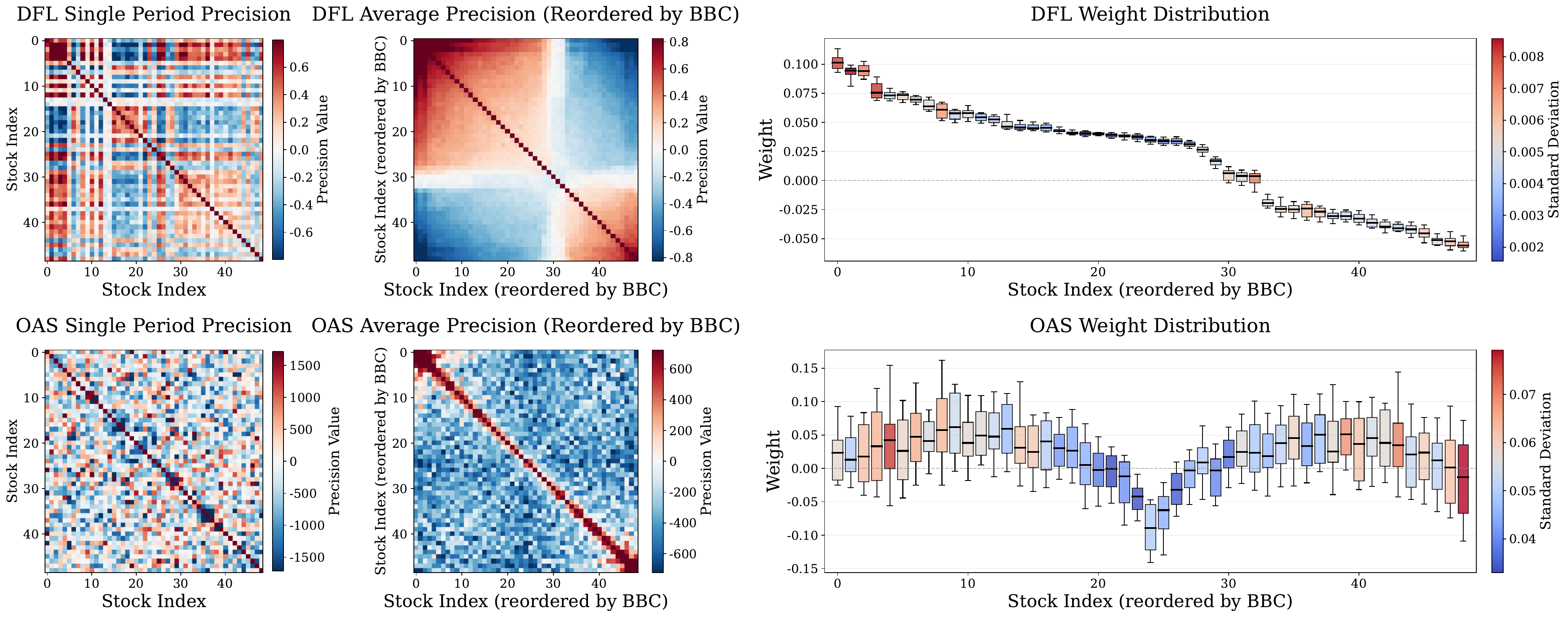}
    \caption{Precision matrix derived by DFL and OAS for single test period in the Industry dataset (left), Average of the permuted precision matrix estimated by DFL and OAS for the entire test period in the Industry dataset (middle), Average and standard deviation of GMVP weights in the reordered index for the entire test period in the Industry dataset (right).}
    \Description{bbc reorder}
    \label{fig:reordered precision}
\end{figure*}


In this section, we investigate the empirical properties of decision‐focused learning (DFL) that contribute to its near‐optimal portfolio decisions. To this end, we compare the GMVP weights produced by DFL with those from other estimators by examining permutations of the predicted precision matrix, and we conduct a performance‐attribution analysis of covariance components to explain DFL’s superior results. Additionally, we suggest key features of DFL which affects to the decision.


Because the analytical solution of the GMVP,
\[
w^{*}(\hat{\Sigma})
= \frac{\hat{\Sigma}^{-1}\mathbf{1}}{\mathbf{1}^\top\hat{\Sigma}^{-1}\mathbf{1}},
\]
which depends on the row sums of \(\hat{\Sigma}^{-1}\), each weight \(w_i\) is proportional to the sum of the \(i\)th row of the precision matrix. 

Figure~\ref{fig:reordered precision} (left) shows the heatmaps of \(\hat{\Sigma}^{-1}\) estimated by DFL and by OAS for a representative test sample. The DFL‐based precision matrix exhibits two distinct groups of rows—one with consistently high values and one with consistently low values—whereas the OAS estimate appears more random and dominated by low values across most rows. We therefore hypothesize that an appropriate row/column permutation of the DFL precision matrix will reveal block structure corresponding to high‐weight and low‐weight stocks.

Algorithm~\ref{alg:bbc} presents the Bidirectional Block Construction (BBC) method for permuting the precision matrix:
\begin{enumerate}
  \item Identify the pair of stocks with the largest off‐diagonal precision value and place them in the top‐left block.
  \item Find the stock with the smallest precision values to this pair and place it in the bottom‐right block.
  \item Iteratively assign remaining stocks with high precision to the top group (left block) or to the bottom anchor (right block), according to their maximum precision with the current blocks.
  \item Continue until all stocks are assigned, forming two coherent blocks.
\end{enumerate}


\begin{algorithm}[t]
\DontPrintSemicolon
\SetAlgoLined
\SetKwInOut{Input}{Input}
\SetKwInOut{Output}{Output}
\Input{Precision matrix $\mathbf{A} \in \mathbb{R}^{n \times n}$}
\Output{Permutation $\pi$, blocks $\mathcal{B}$}
\BlankLine
\tcp{Initialize}
$\text{placed} \leftarrow \{\text{false}\}^n$, $\pi \leftarrow []$\;
$(i^*, j^*) \leftarrow \arg\max_{i < j} |A_{ij}|$\;
$\pi[0] \leftarrow i^*$, $\pi[1] \leftarrow j^*$\;
$\text{placed}[i^*] \leftarrow \text{placed}[j^*] \leftarrow \text{true}$\;
\tcp{Place antipodal anchor}
$k^* \leftarrow \arg\min_{k: \neg\text{placed}[k]} A_{\pi[0],k}$\;
$\pi[n-1] \leftarrow k^*$, $\text{placed}[k^*] \leftarrow \text{true}$\;
\tcp{Bidirectional construction}
$\ell \leftarrow 2$, $r \leftarrow n-2$\;
\While{$\ell \leq r$}{
    $m \leftarrow \arg\max_{k: \neg\text{placed}[k]} \frac{1}{\ell}\sum_{i=0}^{\ell-1} A_{k,\pi[i]}$\;
    $\pi[\ell] \leftarrow m$, $\text{placed}[m] \leftarrow \text{true}$, $\ell \leftarrow \ell + 1$\;
    \If{$\ell \leq r$}{
        $m \leftarrow \arg\max_{k: \neg\text{placed}[k]} \frac{1}{n-r}\sum_{i=r+1}^{n-1} A_{k,\pi[i]}$\;
        $\pi[r] \leftarrow m$, $\text{placed}[m] \leftarrow \text{true}$, $r \leftarrow r - 1$\;
    }
}
$\mathcal{B} \leftarrow \textsc{DetectBlocks}(\mathbf{A}, \pi)$\;
\Return{$\pi, \mathcal{B}$}\;
\caption{Bidirectional Block Construction (BBC) for Precision Matrix Reordering}
\label{alg:bbc}
\end{algorithm}

\begin{figure*}[htbp]
    \centering
    \includegraphics[width=0.9\textwidth]{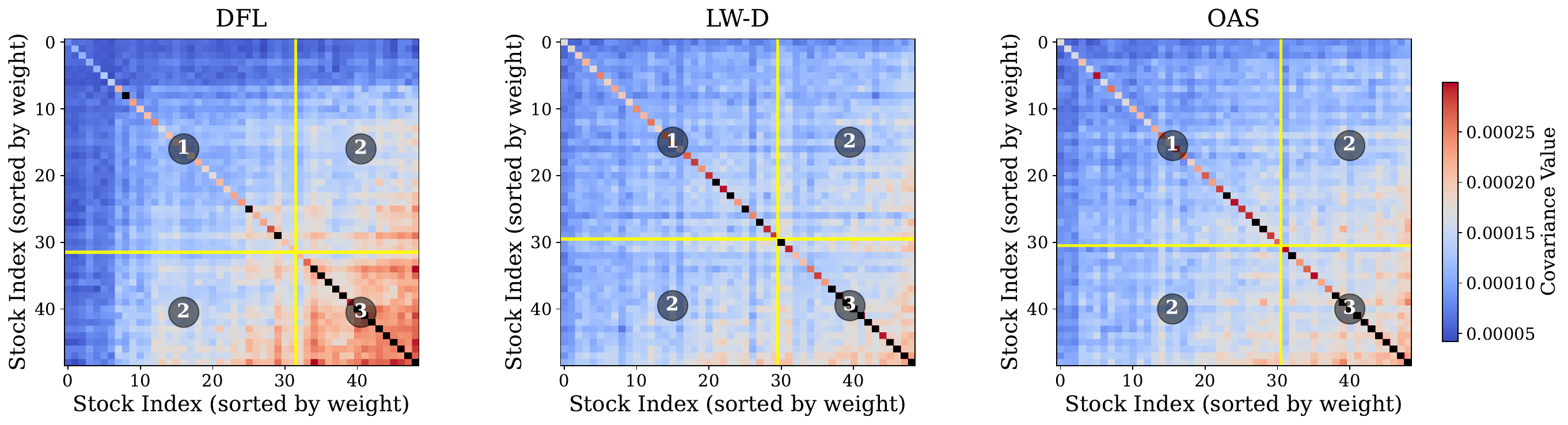}
    \caption{Average true covariance matrix over the test period in the Industry dataset, with assets ordered by descending portfolio weight for DFL (left), LW‑D (center) and OAS (right). The yellow line denotes zero weight: entries to its left or above correspond to positively weighted assets, and those to its right or below to negatively weighted assets. Area \ding{192}: covariance of positively weighted stocks; area \ding{193}: covariance of stocks with mixed weights; area \ding{194}: covariance of negatively weighted stocks. We displayed up to the maximum off-diagonal value, and values exceeding it are displayed in black.}
    \label{fig:weight_cov}
\end{figure*}

\begin{figure}[h]
    \centering
    \includegraphics[width=\columnwidth]{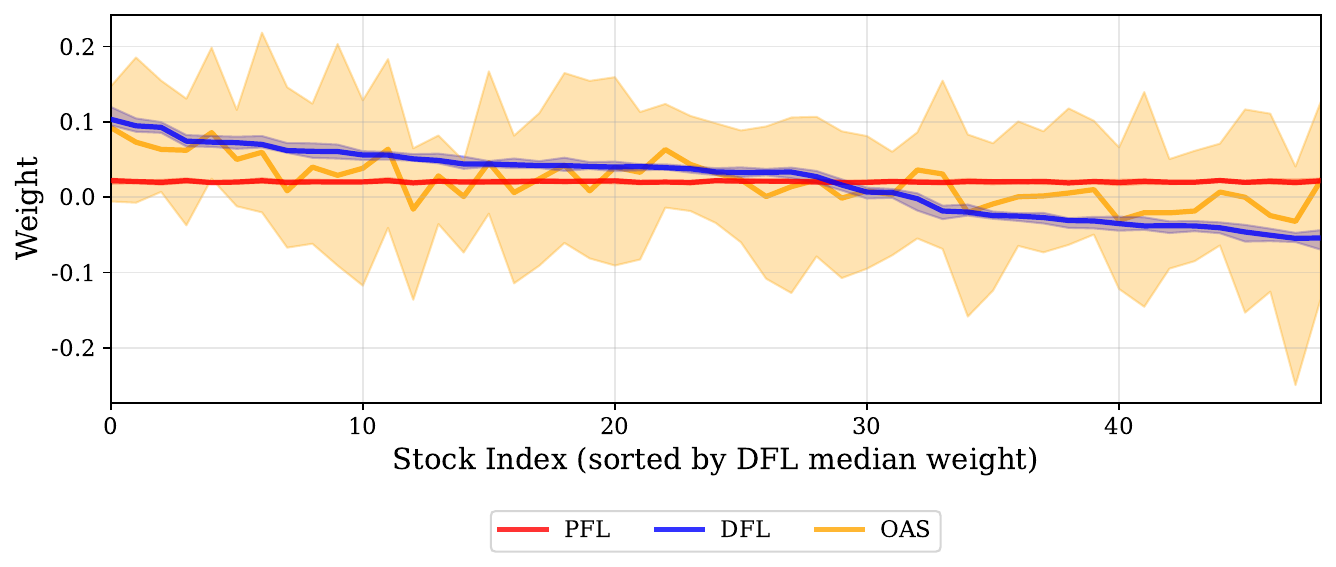}
    \caption{Portfolio weights by the models in whole test period in the Industry dataset. Stock is sorted by the quantity of median weight from DFL. The line indicates median weights of each model and area with same color illustrates one by eliminating weights which are over 97.5 percentile or under 2.5 percentile of whole weights in test period.}
    \Description{Ann vol plot}
    \label{fig:weight_median_percentile}
\end{figure}

The middle panel of figure~\ref{fig:reordered precision} shows the permuted DFL precision matrix, revealing a clear top‐left block of high values and a bottom‐right block of low values. The right panel plots the average GMVP weights computed from the permuted precision matrices over the test period. In contrast, the OAS‐driven permutation yields no discernible block structure and produces weight sequences with high variability. These results indicate that DFL’s precision estimates induce stable, block‐structured weight allocations, whereas OAS yields less consistent decisions over time. Overall, this analysis demonstrates that DFL not only forecasts covariance more effectively for decision making, but also produces a precision‐matrix structure that aligns naturally with robust, low‐volatility portfolio weights. We reported results for S\&P and Dow Jones dataset in figure \ref{fig:reordered precision_appendix}, showing the similar tendency in the Industry dataset.

To identify which assets mainly derive portfolio risk, we decompose the realized volatility of any portfolio \(w\) under the true covariance \(\Sigma_{\mathrm{true}}=[\sigma_{ij}]\in\mathbb{R}^{N\times N}\) as
\begin{align*}
    \mathrm{Vol}_{model}(w) &= w^\top \Sigma_{\mathrm{true}}w\\
    &=w^\top \text{diag}(\sigma_{11},\cdots,\sigma_{NN}) w \\
    &+ w^\top (\Sigma_{\mathrm{true}} - \text{diag}(\sigma_{11},\cdots,\sigma_{NN})) w
\end{align*}
Minimizing the first (diagonal) term requires that assets with small absolute weight \(\lvert w_i\rvert\) coincide with large variances \(\sigma_{ii}\), and vice versa.  
For the second (off‐diagonal) term, volatility is reduced when (1) if $w_{i}w_{j}<0$ and $|w_{i}|, |w_{j}|$ become large, then $\sigma_{ij}$ should become large value or (2) if $w_{i}w_{j}>0$ and $|w_{i}|, |w_{j}|$ become large, then $\sigma_{ij}$ should become small value.

Figure~\ref{fig:weight_cov} presents the test‐period average \(\Sigma_{\mathrm{true}}\), reordered by descending \(\lvert w_i\rvert\).  In region \ding{192}, the DFL‐based model exhibits notably lower true variances and covariances than all other estimators, highlighted to the most heavily weighted assets. Across positively and negatively weighted assets, DFL induces a balancing effect—amplifying risk for extreme positions while attenuating risk for moderate positions—whereas other methods do not illustrate it prominently. However, in the region of short positions (area \ding{194}), DFL’s true covariances remain higher than those of competing estimators, indicating comparatively weaker short‐side risk control. These patterns confirm that DFL primarily reduces overall volatility by selecting low‐variance stocks and moderating large exposures, but may require additional mechanisms to manage short‐position risk. From figure \ref{fig:weight_cov_appendix}, we have checked that our observations work for both S\&P and Dow Jones dataset.

Figure \ref{fig:weight_median_percentile} and \ref{fig:weight_median_percentile_appendix} illustrate, for the test period, each model’s median portfolio weight and the 95\% interval (excluding the top and bottom 2.5\% tails). Unlike the other figures, the stock tickers here are ordered in descending sequence according to the DFL model’s portfolio weights. Two salient observations arise: (1) DFL exhibits clear increase and decrease of long and short positions relative to PFL, and (2) DFL remains consistent allocation patterns which is different with that of OAS. Moreover, both DFL and PFL—which include a training phase—exhibit lower weight dispersion than the other estimators, suggesting that these models capture systematic features during training. We therefore hypothesize that DFL identifies asset characteristics—such as volatility—that drive its allocation decisions.


To evaluate this hypothesis, we rank stocks by their realized volatility of input time-series in whole train period (ascending order) and compare these ranks to the test‐period weight ranks. Table~\ref{tab:volatility_precision} reports the average precision of this correspondence. Unlike the MSE‐based models, DFL achieves high precision for long‐position selection based on low historical volatility. This indicates that DFL systematically has their own criterion which selects low volatile stocks in training phase and reduces portfolio risk by favoring low‐volatility assets. While causality cannot be conclusively established, the results demonstrate that DFL learns and exploits identifiable asset features to construct robust GMVPs.

\begin{table}[htbp]
\centering
\caption{Average precision between top-k volatility rankings (ascending) in training Period and top-k portfolio weight rankings (descending) in test period. Values with the best performances are illustrated in bold.}
\label{tab:volatility_precision}
\begin{tabular}{l|cc|cc|cc}
\hline
\multirow{2}{*}{Model} & \multicolumn{2}{c|}{S\&P 100} & \multicolumn{2}{c|}{49 Industry} & \multicolumn{2}{c}{Dow 30} \\
\cline{2-7}
 & Top 3 & Top 5 & Top 3 & Top 5 & Top 3 & Top 5 \\
\hline  
DFL & \textbf{0.6250} & \textbf{0.7812} & \textbf{0.6667} & \textbf{0.5000} & \textbf{1.0000} & \textbf{0.8000} \\
OAS & 0.0417 & 0.0813 & 0.1562 & 0.2375 & 0.4271 & 0.5438 \\
PFL & 0.0208 & 0.0563 & 0.0833 & 0.0688 & 0.2083 & 0.4062 \\
\hline
\end{tabular}
\end{table}

\section{Conclusion}


In this paper, we have proposed a decision‐focused learning framework for constructing the global minimum‐variance portfolio. 
We have provided a theoretical characterization of DFL by analyzing the principal components of the decision‐aware gradient with respect to the covariance parameter.
Empirically, our results underscore the importance of aligning parameter estimation with the portfolio optimization objective, and demonstrate that conventional MSE‐based methods can lead to suboptimal investment decisions. We show that DFL consistently outperforms both classical estimators and PFL, thereby confirming its superior decision quality. Through a detailed attribution analysis, we further reveal that DFL systematically favors low‐volatility assets, producing robust portfolios with lower realized volatility over time. 


We focus solely on variance minimization and do not assess portfolio returns, since our goal is to examine the learning dynamics and properties of the DFL framework rather than to demonstrate superior return performance. Evaluating conventional metrics (e.g., return, Sharpe ratio) falls outside this paper’s scope. Future work could integrate return objectives to explore DFL’s practical impact on asset allocation.

\bibliographystyle{ACM-Reference-Format}
\bibliography{base}

\clearpage
\onecolumn
\appendix
\section{Appendix}

\FloatBarrier
In this section, we report additional experimental results on different asset universes, including S\&P and  Dow Jones dataset. 
\subsection{Portfolio weights}
To extend the results of portfolio weights for each model in Industry dataset (see figure \ref{fig:ew_pfl}), we checked that our explanation why PFL fails for decision making process works for the different dataset. Figure \ref{fig:ew_pfl_appendix} illustrates portfolio weights of other models and it demonstrates that PFL remains almost same portfolio weights with equally weighted portfolio. 

\begin{figure}[!htbp]
    \centering
    \includegraphics[width=0.95\linewidth]{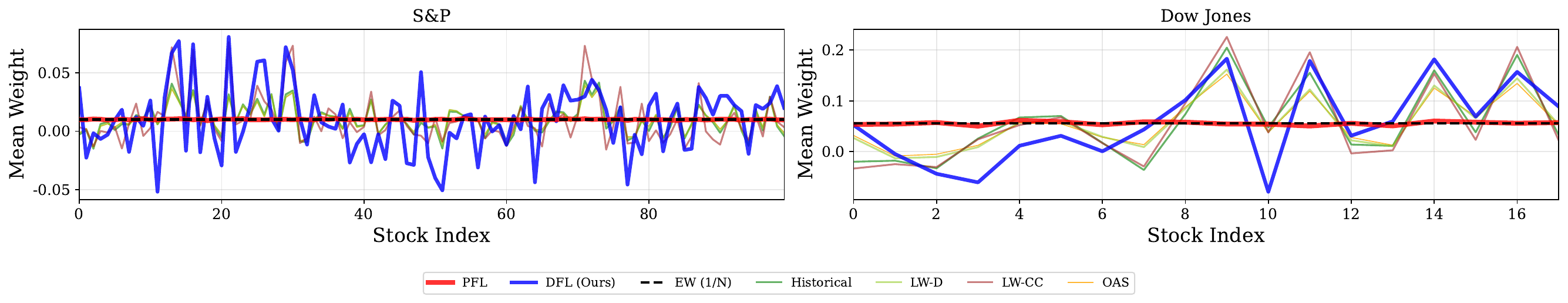}
    \caption{Average portfolio weights in the S\&P and Dow Jones dataset during the test period with $\delta_{in} = \delta_{out} = 21$.}
    \label{fig:ew_pfl_appendix}
\end{figure}

\newpage
\subsection{Precision matrix analysis}
We also check whether precision matrix of DFL have similar structures as depicted in figure \ref{fig:reordered precision}, algorithm \ref{alg:bbc} works for sorting precision matrix and weight distribution for DFL and OAS model. Figure \ref{fig:reordered precision_appendix} shows that our hypothesis and algorithm hold for different asset universes, explaining similar results of figure \ref{fig:reordered precision}.
\begin{figure}[!htbp]
    \centering
    \includegraphics[width=0.95\linewidth]{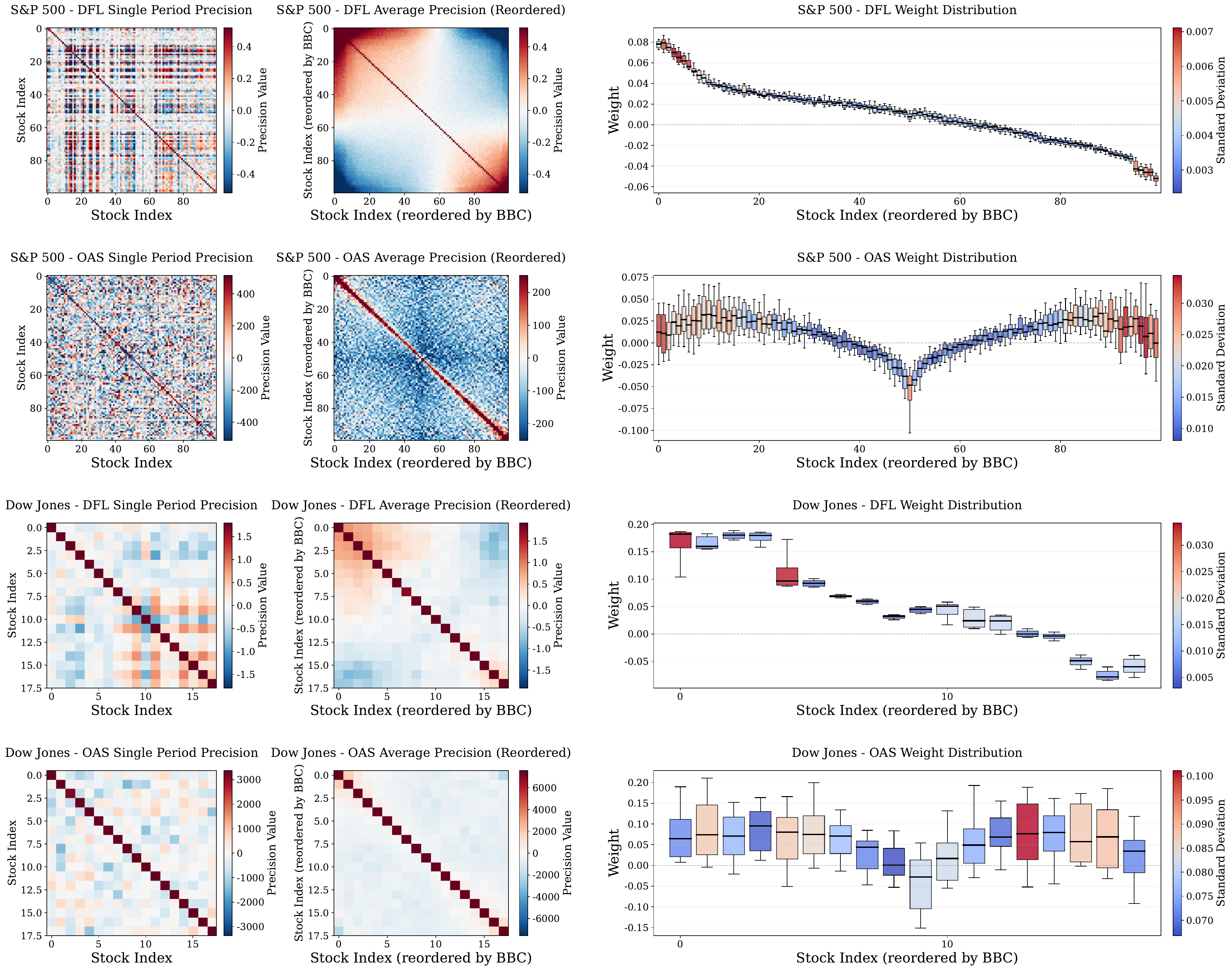}
    \caption{Precision matrix derived by DFL and OAS for single test period (left), average precision matrices reordered by BBC algorithm (middle), and weight distributions (right) for S\&P and Dow Jones datasets.}
    \label{fig:reordered precision_appendix}
\end{figure}

\newpage
\subsection{Performance attribution of minimizing variances}
We employ the performance attribution in the variance minimization problem with estimated covariance matrix. Figure \ref{fig:weight_cov_appendix} shows true covariance sorted by the order of weights in each model. Similar to results of figure \ref{fig:weight_cov}, It shows that stocks index whose portfolio weights of DFL are large have relatively small true variance and covariance values.

\begin{figure}[!htbp]
    \centering
    \includegraphics[width=0.90\linewidth]{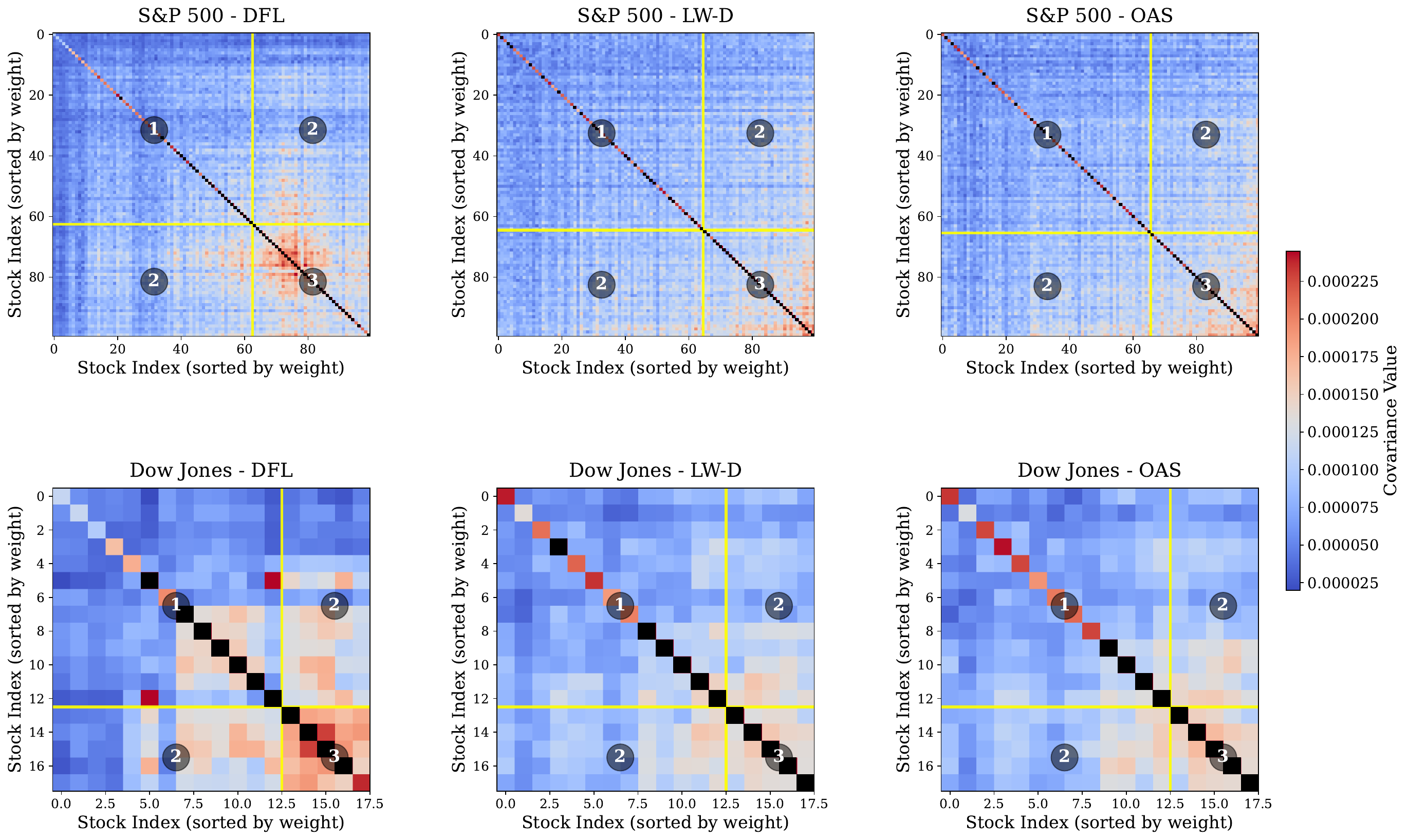}
    \caption{Average true covariance matrix over the test period in the Industry dataset, with assets ordered by descending portfolio weight for DFL (left), LW- D (center) and OAS (right) in S\&P 500 and Dow Jones datasets.}
    \label{fig:weight_cov_appendix}
\end{figure}

\subsection{Invariance of stock selection in DFL}
We have show that each model’s median portfolio weight ordered in descending sequence according to the DFL model’s
portfolio weights which is an extension of figure \ref{fig:weight_median_percentile} to the other asset universes. Figure \ref{fig:weight_median_percentile_appendix} shows similar tendency, demonstrating that consistent asset allocation of DFL holds. (While the result of S\&P shows slightly large confidence interval compared with results in other asset universes, it is not compatible with the results of other estimator except PFL.)
\begin{figure}[!htbp]
    \centering
    \includegraphics[width=0.90\linewidth]{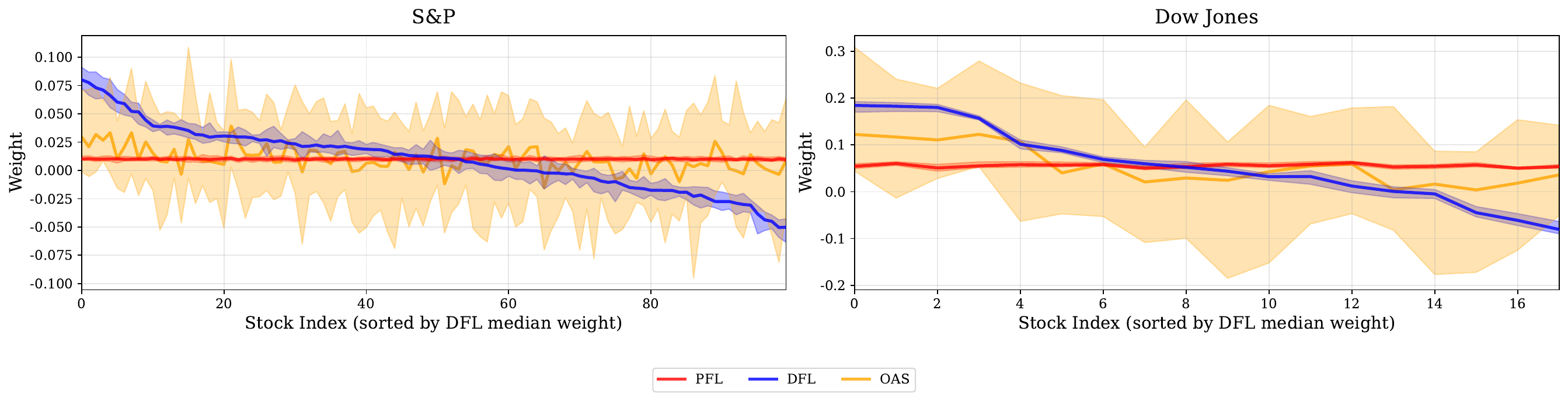}
    \caption{
    Portfolio weights by different models in the test period for S\&P 500 and Dow Jones datasets. Stock indices are sorted by the quantity of median weight from DFL. Lines indicate median weights and shaded areas illustrate one by eliminating weights which are over 97.5 percentile or under 2.5 percentile of whole weights in test period.}
    \label{fig:weight_median_percentile_appendix}
\end{figure}

\end{document}